\documentclass{amsart}
\usepackage{graphicx} % Required for inserting images
\usepackage{helvet}
\usepackage{amsmath}
\usepackage{amssymb}
\usepackage{amsthm}
\usepackage{mathtools}
\usepackage[svgnames]{xcolor}
\usepackage{graphicx}
\graphicspath{{./images/}}
\usepackage{geometry}
\usepackage{datetime}
\usepackage[shortlabels]{enumitem}
\usepackage{dirtytalk}
\usepackage[normalem]{ulem}
\usepackage{float}
\restylefloat{table}
\usepackage{placeins}

\usepackage[hidelinks]{hyperref}
\usepackage[biblabel]{cite}
\usepackage{tikz}
\usepackage{orcidlink} % requires hyperref, tikz
\usepackage[foot]{amsaddr}
\newcommand{\F}{\mathbb{F}}
\newcommand{\Z}{\mathbb{Z}}

\newcommand{\fqs}{\F_{q^2}}
\newcommand{\ev}{\text{ev}}

\newcommand{\cF}{\mathcal{F}}
\newcommand{\cD}{\mathcal{D}}
\newcommand{\mon}{X^{e_1}Y^{e_2}}
\newcommand{\monp}{X^{e_1'}Y^{e_2'}}
\newcommand{\expair}{((e_1,e_2),(e_1',e_2'))}

\newcommand{\specialcell}[2][c]{%
\begin{tabular}[#1]{@{}c@{}}#2\end{tabular}}
\newtheorem{theorem}{Theorem}[section]
\newtheorem{corollary}[theorem]{Corollary}
\newtheorem{lemma}[theorem]{Lemma}

\newtheorem{proposition}[theorem]{Proposition}

\theoremstyle{definition}
\newtheorem{exmp}{Example}[section]

\theoremstyle{definition}
\newtheorem{definition}[theorem]{Definition}

\theoremstyle{remark}
\newtheorem{remark}[theorem]{Remark}

\title[New Quantum Stabilizer Codes from Generalized Monomial Cartesian Codes]{New Quantum Stabilizer Codes from Generalized Monomial Cartesian Codes
constructed using two  different generalized Reed-Solomon codes.}
\author[O. Campion]{Oisin Campion\orcidlink{0009-0004-1951-7103}}
\address[Oisin Campion]{School of Mathematics and Statistics, University College Dublin, Ireland}
\email{oisin.campion@ucdconnect.ie}

\author[F. Hernando]{Fernando Hernando$^*$  \orcidlink{0000-0002-9758-2152}}
\address[Fernando Hernando]{Department of Mathematics and Institut Universitari de Matemàtiques i Aplicacions de Castelló,
Universidad Jaume I, Spain}
\email{carrillf@mat.uji.es$^*$}
\author[G. McGuire]{Gary McGuire \orcidlink{0000-0003-2105-9792}}
\address[Gary McGuire]{School of Mathematics and Statistics, University College Dublin, Ireland}
\email{gary.mcguire@ucd.ie}

\thanks{This publication has emanated from research conducted with the financial support of Science Foundation Ireland under Grant number 21/RP-2TF/10019 for the first author. The second author has been partially supported by  grant  PID2022-138906NB-C22
funded by MCIN/AEI/ 10.13039/501100011033 and ERDF, UE and by grant GACUJIMA-2024-03 funded by Universitat Jaume I.}
\keywords{Generalised Monomial Cartesian Code, Stabilizer Code, MDS code, Gilbert-Varshamov bound}
\subjclass[2020]{81P70, 94B05, 14G50, 11T71}

\begin{document}
\begin{abstract}
    In this work, we define Generalized Monomial Cartesian Codes (GMCC), which constitute a natural extension of generalized Reed–Solomon codes. We describe how two different generalized Reed-Solomon codes can be combined to construct one GMCC. We further establish sufficient conditions ensuring that the GMCC are Hermitian self-orthogonal, thus leading to new constructions of quantum codes. 
\end{abstract}
\maketitle

\section{Introduction}
Quantum computing has the potential to fundamentally transform computation by exploiting uniquely quantum phenomena such as \emph{superposition} and \emph{entanglement}. These features enable certain quantum algorithms to outperform the best known classical algorithms. A paradigmatic example is Shor’s algorithm for integer factorization, which runs in polynomial time on a quantum computer but for which no classical polynomial‑time algorithm is known, with profound implications for cryptography and complexity theory \cite{Shor1995Factoring}. Another fundamental quantum algorithm is Grover’s search algorithm, which achieves a quadratic speedup over classical search methods \cite{Grover1997}.  

Despite this promise, quantum information is extremely fragile due to decoherence and noise arising from interactions with the environment. Unlike classical information, quantum information cannot be copied due to the \emph{no‑cloning theorem}, which imposes fundamental limits on how errors can be mitigated \cite{NoCloning}. Moreover, quantum errors are more complex than classical bit flips, as they can affect both amplitude and phase.  

To address these challenges, the theory of quantum error correction was developed in the mid‑1990s. The first quantum error‑correcting code was introduced by Shor, who showed that a logical qubit could be protected against arbitrary single‑qubit errors by encoding it into nine physical qubits and correcting for noise actively \cite{Shor1995}. Shortly thereafter, Calderbank and Shor, as well as Steane, independently developed the class of quantum error‑correcting codes now known as CSS codes, which systematically construct quantum codes from pairs of classical linear codes \cite{CalderbankShor1996,Steane1996}. These constructions demonstrated that reliable quantum information processing is theoretically possible even in the presence of noise, laying the foundation for fault‑tolerant quantum computation.  

The stabilizer formalism, introduced by Gottesman, provides a unifying algebraic framework for describing many quantum codes, including CSS codes and more general constructions \cite{Gottesman1997}. The existence of quantum codes with nontrivial distance and reasonable rates was further explored using orthogonal geometry and additive codes over finite fields \cite{CalderbankRains1996}. Together, these foundational results show that quantum error correction is essential to practical quantum computing: without codes capable of detecting and correcting errors while preserving quantum coherence, the advantages offered by quantum algorithms cannot be realized in practice.

Among quantum codes, \emph{quantum MDS (Maximum Distance Separable) codes} are of particular interest because they achieve the quantum Singleton bound, which gives the maximum possible distance for a given length and dimension \cite{Grassl2004}.  Quantum  MDS codes provide the highest error-correcting capability for their parameters, making them optimal in this sense. However, constructing new MDS quantum codes is notoriously difficult and it is limited to lengths smaller than $q^2+1$ if MDS conjecture is true, so extending them to larger lengths or more general alphabets is a significant challenge. Recent work by Campion et al. \cite{campio2025} presented new families of quantum MDS codes using advanced algebraic techniques, expanding the known landscape but also highlighting the limitations of current constructions. This motivates the exploration of alternative families of codes that are easier to generalize while still supporting quantum error correction, such as the constructions we propose here.

In this paper, we introduce a class of \emph{generalized monomial–Cartesian (GMC) codes}. We analyze the conditions under which these codes are contained in their Hermitian dual, enabling the systematic construction of quantum stabilizer codes through the Hermitian construction \cite{kkks}. Our approach produces families of quantum codes whose parameters exceed the quantum Gilbert–Varshamov bound, and among them we identify explicit examples that improve upon all previously known quantum codes in the literature for their respective lengths and dimensions.

The paper first introduces the necessary preliminaries in Section 2. Section 3 explains how to combine two GRS codes to construct a Generalized Monomial Cartesian Code. Section 4 focuses on the special case of selecting only two points in the second variable, while Sections 5 and 6 study the Hermitian self-orthogonality conditions for codes with  several points  and compares the results with the Gilbert-Varshamov bound. Section 7 presents families of codes derived from these constructions, and Section 8 provides explicit examples along with a comparison to existing results in the literature.

\section{Preliminaries}

Let $q$ be a prime power. First, we introduce generalized Reed-Solomon codes defined over $\F_{q^2}$. We fix some arbitrary subset $P=\{P_1,\ldots,P_n\}\subseteq\fqs$. Let $V_k$ be the vector space of all polynomials in $\fqs[X]$ of degree $<k<n$. Let $v \in (\fqs^*)^n$.

\begin{definition}
    A \textbf{generalized Reed-Solomon code} (GRSC) $\text{GRS}_{k}(v, P)$ is obtained as the image of the injective evaluation map: 
\[
\text{ev}_{v,P}:V_k\rightarrow \fqs^n, ~~\ev_{v,P}(f)=(v_1f(P_1),\ldots , v_nf(P_n)).
\]
\end{definition}

A GRS code has parameters $[n,k,n-k+1]_{q^2}$. A GRS code is always MDS, and the dual of a GRS code is another GRS code. 

There is a natural extension of GRS codes to codes defined over bivariate polynomial spaces, based on generalized Monomial–Cartesian Codes (MCC). MCC were introduced in \cite{Hiram2020}, and they can also be viewed as a particular subfamily of affine variety codes \cite{Fitz1998}. Moreover, MCC are of independent interest because of their flexible structure, which allows for various applications; for instance, they admit an efficient locally recoverable decoding algorithm \cite{Gal2024} and can be used to construct CSS-T codes with good parameters \cite{Seyma2025}.

Let $P_X,P_Y \subset\fqs$ be arbitrary subsets, where $n_X=|P_X|, n_Y=|P_Y|$ and let $I \leq \fqs[X,Y] $ be the vanishing ideal of $P:=P_X \times P_Y$. 
Let $n= |P|$ and fix some ordering on the elements of $P$, writing $P=\{P_1,\ldots,P_n\}$. Let $E$ be the set of monomial exponents in the coordinate ring  $\fqs[X,Y]/I$ and let $\Delta\subset E$.

Let $V \subset \mathbb{F}_{q^2}[X,Y]/I$ be the $|\Delta|$-dimensional vector space spanned by the monomials in $\Delta$, and let  
$Q = (Q_1,\ldots,Q_n) \in (\mathbb{F}_{q^2}^{*})^{n}$.

\begin{definition}\label{d:GMCC}
A \textbf{generalized monomial Cartesian code} $\mathrm{GMC}_{\Delta}(Q,P)$ is the image of the injective evaluation map
\[
\mathrm{ev}_{Q,P} : V \longrightarrow \mathbb{F}_{q^2}^{\,n},
\qquad 
\mathrm{ev}_{Q,P}(f) = (Q_1 f(P_1),\ldots,Q_n f(P_n)).
\]
Equivalently,
\[
\mathrm{GMC}_{\Delta}(Q,P)
    = \operatorname{Span}\bigl\{
        \mathrm{ev}_{Q,P}(X^{a} Y^{b})
        \,\bigm|\,
        (a,b)\in \Delta
      \bigr\}
    \subset \mathbb{F}_{q^2}^{\,n}.
\]
\end{definition}

Such codes have length $n$ and dimension $ |\Delta|$.   To estimate the minimum distance we use the following version of the \textbf{footprint bound} (see \cite{Geil2000}).

\begin{definition}
\[
\delta_{FB}\!\left(\mathrm{GMC}_{\Delta}(Q,P)\right)
    = 
    \min_{(a,b)\in \Delta}
    (n_X - a)(n_Y - b).
\]
\end{definition}

The footprint bound provides the usual lower bound for the minimum distance.

\begin{proposition}
\[
d\!\left(\mathrm{GMC}_{\Delta}(Q,P)\right)
    \;\ge\;
    \delta_{FB}\!\left(\mathrm{GMC}_{\Delta}(Q,P)\right).
\]
\end{proposition}

\begin{theorem}\label{Te:EqualityIf}
If $\mathrm{GMC}_{\Delta}(Q,P)$ is a hyperbolic code \cite{Geil2001} or a 
monomial-decreasing Cartesian code \cite{Garcia2020}, then  
\[
d\!\left(\mathrm{GMC}_{\Delta}(Q,P)\right)
    \;=\;
    \delta_{FB}\!\left(\mathrm{GMC}_{\Delta}(Q,P)\right).
\]
\end{theorem}

Let $t>1$ be a positive integer and consider the following set of monomials:
\[
\Delta_t = \{\,X^{a}Y^{b} \in \mathbb{F}_{q^2}[X,Y] : (a+1)(b+1) < t\,\}.
\]
To simplify notation, we denote $\mathrm{GMC}_{\Delta_t}(Q,P)$ simply by $\mathrm{GMC}_t(Q,P)$. 
It is well known that $\mathrm{GMC}_t(Q,P)$ is a monomial-decreasing Cartesian code and that its dual is a hyperbolic code. Therefore, as a consequence of Theorem~\ref{Te:EqualityIf}, we obtain the following result.

\begin{corollary}
The minimum distance of the dual of $\mathrm{GMC}_t(Q,P)$ is $t$.
\end{corollary}

Now we introduce some elementary concepts and results that we will use throughout the paper (see \cite{NielsenChuang2010}).

A \textbf{quantum stabilizer code} is the common $+1$-eigenspace of a commutative subgroup of Pauli operators, which we call the \textbf{stabilizer group}. Stabilizer codes can be constructed using classical linear codes defined over $\mathbb{F}_{q^2}$ that are self-orthogonal with respect to a Hermitian inner product. For a detailed treatement of non-binary quantum stabilizer codes, see \cite{kkks}.

\begin{theorem}\label{t:HermCon}
Let $C$ be a linear $[n,k,d]_{q^2}$ error-correcting code over the field $\mathbb{F}_{q^2}$ such that $C \subseteq C^{\perp_h}$. Then, there exists an $[[n,n-2k,\geq d^{\perp_h}]]_q$ quantum stabilizer code, where $d^{\perp_h}$ denotes the minimum distance of $C^{\perp_h}$.
\end{theorem}

\begin{corollary}\label{c:GMCCParams}
Let $\mathrm{GMC}_t(Q,P)$ be a GMC code defined over $\mathbb{F}_{q^2}$ such that 
$\mathrm{GMC}_t(Q,P) \subseteq (\mathrm{GMC}_t(Q,P))^{\perp_h}$. Then there exists an $[[n,n-2|\Delta_t|,t]]_q$ quantum stabilizer code. 
\end{corollary}

Our objective in this paper is to generalize the result in \cite{campio2025}, namely, to construct GMC codes in two variables that satisfy the conditions of Theorem~\ref{t:HermCon}, thereby enabling the construction of quantum error-correcting codes. In the next section, we will show how two GRS codes can be combined to obtain a GMC code.

\section{How to turn two GRS codes into one GMC code }\label{s:GMCframework}

The foremost difficulty in the application of classical error-correcting codes to quantum error-correction is the imposition of a self-orthogonality condition, such as in Theorem \ref{t:HermCon}. In the case of GMC codes defined over $\fqs$, a convenient basis for the space of codewords is: 
\[
\mathcal{B}=\{\text{ev}_{Q,P}(f): f \in \Delta_t\}. 
\]
Our codes will be Hermitian self-orthogonal if and only if every basis element is orthogonal to every other basis element; in other words, we require that 
\[
\text{ev}_{Q,P}(f)\cdot_h\text{ev}_{Q,P}(g)=0 ~~\forall f,g \in \Delta_t.
\]
Since every monomial is uniquely determined by its exponent, we consider the set of all pairs of monomial exponents, whose corresponding evaluation vectors are not orthogonal. 

\begin{definition}
    For a given GMC code, we define the set of \textbf{non-orthogonal exponent pairs} to be 
    \[
    \mathcal{D}=\{((e_1,e_2),(e_1',e_2')) \in E\times E: \text{ev}_{Q,P}(X^{e_1}Y^{e_2})\cdot_h \text{ev}_{Q,P}(X^{e_1'}Y^{e_2'}) \neq 0\}. 
    \]
\end{definition}

\begin{remark}
    Since there is a correspondence between monomials and their exponents, we will often talk of membership of monomials in the sets $E,\mathcal{D}$, rather than saying ``the exponents of the monomials''. Later, we will frequently refer to a pair of monomials in $E\times E$ as a single ``point''. 
\end{remark}

\begin{remark}\label{r:deltat}
    Our code will be Hermitian self-orthogonal if and only if no pair of monomials from $\Delta_t$ lies in $\mathcal{D}$. Since the sets $\Delta_t$ are nested, of particular interest is determining the maximum value of $t$ for which this happens. 

\end{remark}

The characterization of $\mathcal{D}$ depends on the choice of evaluation sets $P_X=\{x_i\},P_Y=\{y_i\}$ as well as the coordinate vector $Q$. Having a complete characterization is difficult. Instead, we will find some other set $\mathcal{F} \supseteq \mathcal{D}$, and it will be sufficient to check that we have no pairs of monomials from $\Delta_t$ lying in $\mathcal{F}$.

 Let us now look in more detail at the Hermitian inner product. We will write the evaluation points as $P(i,j)=(x_i,y_i)$ and coordinate vector as $Q(i,j)$. Given two monomials $X^{e_1}Y^{e_2}$ and $X^{e_1'}Y^{e_2'}$, then the Hermitian inner product between them is given by the expression: 
\begin{equation*}
\begin{split}
 &ev_{Q,P}\left(X^{e_1}Y^{e_2}\right)\cdot_h ev_{Q,P}\left(X^{e_1'}Y^{e_2'}\right) \\
 = &\sum_{i,j}Q(i,j)^{q+1}(x_i)^{e_1+qe_1'}(y_j)^{e_2+qe_2'}.
\end{split}
\end{equation*}
If it so happens that the coordinate vector can be written as $Q(i,j)=v_i\cdot w_j$ then this sum can be factored as: 
\begin{equation}\label{eq:decomp}
\begin{split}
  &\sum_{i,j}Q(i,j)^{q+1}(x_i)^{e_1+qe_1'}(y_j)^{e_2+qe_2'} \\
        =&\left(\sum_i v_i^{q+1}(x_i)^{e_1+qe_1'}\right)\left( \sum_jw_j^{q+1}(y_j)^{e_2+qe_2'} \right)\\
        =&   \left( ev_{v,P_X}\left(X^{e_1}\right)\cdot_h ev_{v,P_X}\left(X^{e_1'}\right)\right)\cdot \left(ev_{w,P_Y}\left(Y^{e_2}\right)\cdot_h ev_{w,P_Y}\left(Y^{e_2'}\right) \right).
\end{split}
\end{equation}

where $\text{ev}_{v,P_X}$ (resp. $\text{ev}_{w,P_Y}$) is the evaluation map corresponding to the GRS code $\text{GRS}_{k}(v, P_X)$ (resp. $\text{GRS}_{k}(w, P_Y)$). In this way, we can see that we can combine two GRS codes to create a GMC code. Moreover, the orthogonality of codewords of the GMC code is completely determined by the orthogonality of codewords of each GRS code. 

\begin{definition}
    If the coordinate vector $Q(i,j)$ of a GMC code can be decomposed as $Q(i,j)=v_i\cdot w_j$ then we say that the GMC code is \textbf{separable}. 
\end{definition}

\begin{remark}
    We will restrict our study of GMC codes to those that are separable. Otherwise, the factorization in Equation \ref{eq:decomp} does not hold. 
\end{remark}

\begin{remark}\label{r:D_X}
    As we did for GMC codes, we can define for each of the GRS codes the set of \textbf{non-orthogonal exponent pairs}. More precisely, we define
    \[
    \mathcal{D}_X=\{(e_1,e_1'):\text{ev}_{v,P_X}(X^{e_1})\cdot_h\text{ev}_{v,P_X}(X^{e_1'})\neq 0\}. 
    \]
    We define $\cD_Y$ similarly.
\end{remark}

\begin{proposition}\label{p:decomp}
    For any pair of monomial exponents $((e_1,e_2),(e_1',e_2')) \in E\times E$, we have that 
    \[
    ((e_1,e_2),(e_1',e_2')) \in \cD \iff (e_1,e_1')\in \cD_X \text{ and } (e_2,e_2')\in \cD_Y.
    \]
\end{proposition}

\begin{proof}
    This is immediate from Equation \ref{eq:decomp}.
\end{proof}
\begin{remark}
    As we mentioned earlier, a full characterization of $\cD_X$ and $\cD_Y$ will be difficult. Instead we will find suitable sets $\cF_X\supseteq\cD_X$ and $\cF_Y\supseteq \cD_Y$. 
\end{remark}
In this paper, we will fix a class of GRS codes for use in the $X$-variable (that is, fixing $P_X$ and $v$). This family of codes was introduced in \cite{campio2025}, in which the authors derived conditions for the Hermitian-orthogonality of codewords. In later sections, we will provide an explicit description of another family of GRS codes for use in the $Y$ variable, (that is, selecting $P_Y$ and $w)$ which in combination with the $X$-variable codes will provide us with Hermitian self-orthogonal GMC codes. 

Now we recall the first family of GRS code, for use in the $X$-variable. 

We assume that $q\ge 4$. Let $\lambda>1$ be a divisor of $q-1$ and let $\tau,\rho>1$ be divisors of $q+1$. We assume that $\gcd(\lambda,\tau)=1$ and that $\frac{\rho}{\gcd(\lambda\tau,\rho)}\ge2.$ We choose $\sigma$ to be any integer with $2 \leq \sigma \leq \frac{\rho}{\gcd(\lambda\tau,\rho)}$. We consider the set of evaluation points:
\[
P_X=\{\zeta_\lambda^{i}\zeta_{\tau}^{j}\zeta_{\rho}^{\ell} : 
    0\leq i < \lambda, \ 0\leq j < \tau,\
    0\leq \ell < \sigma\} \subset \F_{q^2}
\]
where $\zeta_t$ denotes a primitive $t$-th root of unity. 
By Lemma 3.1 of \cite{campio2025}, the elements of $P_X$ are distinct, and we can uniquely associate triples $(i,j,\ell)$ with elements of $P_X$, writing $P_X(i,j,\ell):=\zeta_\lambda^{i}\zeta_{\tau}^{j}\zeta_{\rho}^{\ell}$. 

Next, we describe the construction of the vector $v$. We select $s_0,\ldots s_{\sigma-1}\in \F_q^*$ in the following way: 

\begin{itemize}
    \item If $\sigma=2$, then set $s_0=1, s_1=-1$. 
    \item Otherwise, set $s_0,\ldots,s_{\sigma-3} =1$, select $s_{\sigma-2} \in \F_q \setminus \{0,-(s_0+\ldots +s_{\sigma-3})\}$ and set $s_{\sigma-1} = -(s_0+\ldots s_{\sigma-2})$. This requires the the assumption that $q>2$.
\end{itemize}

This ensures that $\sum_{\ell=0}^{\sigma-1}s_\ell=0$. Let $L$ be an integer, chosen according to Table \ref{tab:valueL}. We will label the coordinates of the vector $v$ by $(i,j,\ell)$, 
using the same ordering as the evaluation set $P_X$.
We define $v$ by
$v(i,j,\ell)^{q+1}=\zeta_\lambda^{-iL} s_{\ell}.$
Since $\zeta_\lambda^{-iL} s_{\ell}$ is an element of $\F_q$,
it is always possible to solve this norm equation and find a suitable $v$ with coefficients in $\F_{q^2}$. 

\begin{definition}
    With these parameters defined, we define the code
    \[
    C_X:= GRS_{k}(v,P_X).
    \]
\end{definition}
Let us now describe the orthogonality properties of this code, which were proved in \cite{campio2025}.
\begin{proposition}(Theorem 3.4 of \cite{campio2025})\label{p:orthogconditions}
    Let $X^{e_1},X^{e_1'}$ be two monomials. Then
    $$
    ev_{v,P_X}\left(X^{e_1}\right)\cdot_h ev_{v,P_X}\left(X^{e_1'}\right) =0
    $$
    if  
    any one of the following conditions holds: 
    \begin{itemize}
    \item $e_1+e_1'  \not\equiv L $ (mod $\lambda$).
    \item $e_1 \not\equiv e_1'$ (mod $\tau$).
    \item $e_1\equiv e_1'$ (mod $\rho$).
\end{itemize}
\end{proposition}

\begin{definition}\label{d:Xfailurepoint}
    Let $X^{e_1},X^{e_1'}$ be two monomials. We denote by $\mathcal{F}_X$ the set of non-negative integer pairs $\left(e_1,e_1'\right)$ with $e_1\not= e_1'$ satisfying both of the following conditions: 
\begin{enumerate}
    \item $       e_1+e_1'  \equiv L  \textnormal{ (mod } \lambda).$
    \item $  e_1  \equiv e_1'  \textnormal{ (mod } \tau).$
\end{enumerate}
    We call elements of the set $\cF_X$ \textbf{$X$-failure points}. 
\end{definition}

\begin{remark}\label{r:Xsymmetry}
Notice that if a point $(a,b)$ satisfies both parts of Definition \ref{d:Xfailurepoint} then so does $(b,a)$. We usually assume without loss of generality that a point $(e_1,e_1')\in\cF_X$ satisfies $e_1<e_1'$.
 \end{remark}

\begin{lemma}
    The set $\cF_X$ contains $\cD_X$.
\end{lemma}
\begin{proof}
    This follows directly from Remark \ref{r:D_X}, Definition \ref{d:Xfailurepoint} and Proposition \ref{p:orthogconditions}. 
\end{proof}
\begin{remark}
    We could sharpen this containment by considering also the third condition in Definition \ref{d:Xfailurepoint}, but this adds a lot of complexity for relatively little gain. In certain specific cases we can use the condition to construct some more codes (see the second paragraph of Example \ref{exmp:codeexample}).
\end{remark}
\begin{remark}
    The reason for calling the elements of $\cF_X$ $X$-failure points is the following. Given two monomials $\mon$ and $\monp$, their corresponding evaluation vectors will be orthogonal if and only if $ (e_1,e_1')\notin \cD_X$ or $(e_2,e_2')\notin \cD_Y$, by Proposition \ref{p:decomp}. If $(e_1,e_1') \notin \cF_X$, then $(e_1,e_1') \notin \cD_X$ and we can determine immediately that the evaluation vectors of the monomials are orthogonal. Otherwise we are risk of ``failing'' orthogonality. 
\end{remark}

\begin{exmp}
        Let $q=11,\lambda=5,\tau=3,\rho=4$. Then since $(3,5)\notin \cF_X$ (this point does not satisfy part 2 of Definition \ref{d:Xfailurepoint}), we can guarantee that the evaluation vectors of the monomials $X^3,Y^{e_2}$ and $X^5,Y^{e_2'}$ are orthogonal for any $e_2,e_2'$.
\end{exmp}

\begin{lemma}\label{l:latticemovement}
   Every $X$-failure point with $e_1<e_1'$ can be written uniquely as 
    \begin{equation}\label{eq:writingXpoint}
          (e_1,e_1') = (T_1,T_2)+i(\lambda/2,\lambda/2)+j(-\tau/2,\tau/2)  
    \end{equation}
where $i,j\in \Z_{\ge 0}$, and $(T_1,T_2)$ is given by Table \ref{tab:valueL}, with the exception of the case with $\lambda ,\tau$ odd, $\rho=2$ and $\lambda \ge \tau+2$. 
This exception can be safely ignored, see Remark \ref{ignore1}.
\end{lemma}

\begin{proof}
    Using Propositions 5.4 and 4.13 of \cite{cam-jose}, it follows that $i\ge 0$, and in Cases 2 and 3 from Table1 \ref{tab:valueL} that $j \ge 0$. In Case 1 from Table \ref{tab:valueL}, we observe that $T_2-T_1=2\tau$, and it follows from parity there there are no points with $T_2-T_1=\tau$, whence $j \ge 0$ by Proposition 4.13 of \cite{cam-jose}. 
\end{proof}

\begin{remark}\label{ignore1}
    In the case with $\lambda, \tau$ odd, $\rho=2$ and $\lambda \ge2$, it follows from Proposition 5.4 of \cite{cam-jose} that the only exceptional points satisfy part $3$ of Proposition \ref{p:orthogconditions}, and therefore the monomials are orthogonal and can be ignored. 
\end{remark}

    \begin{table}[h!]
        \centering
\begin{tabular}{|| c|c| c ||} 
 \hline
Conditions & Value of $L$ &$(T_1,T_2)$\\  
 \hline\hline
$\lambda$ even 
&$2\tau-2$
 & $\left(\frac{\lambda-2}{2},\frac{\lambda+4\tau-2}{2}\right)$\\ 
 \hline
\specialcell{$\lambda$ odd, and at least one of the following\\
 conditions holds: $\lambda<\tau$,  $\tau$ even, $\rho=2$}
&$\tau-2$
 & $\left(\lambda-1,\lambda+\tau-1\right)$    \\
 \hline
$ \lambda $ odd, $\lambda >\tau, \tau $ odd, $\rho \neq 2$ 
&$2\tau-2$
 & $\left(\frac{\lambda+\tau-2}{2},\frac{\lambda+3\tau-2}{2}\right)$  \\
 \hline
\end{tabular}
        \caption{Values of $L$ and $(T_1,T_2)$}
        \label{tab:valueL}
    \end{table}

We now come to the task of determining a suitable GRS code for use in the $Y$ variable. The structure of this code should complement the structure of our first GRS code, in order to ensure Hermitian self-orthogonality for a wide range of parameters. To motivate our choice of code, we introduce the notion of the \textbf{total footprint}. 

Recall from Remark \ref{r:deltat} that our goal is to determine for which value of $t$ can we ensure that no pair of monomials from $\Delta_t$ lies in $\cD$. For each given pair $\expair \in \cD$, the smallest value of $t$ such that both $\mon$ and $\monp \in \Delta_t$ is $\max\{(e_1+1)(e_2+1), (e_1'+1)(e_2'+1)\}$. This motivates the following definition

\begin{definition}
    Given a pair of monomial exponents $((e_1,e_2),(e_1',e_2')) \in E\times E$, we define the \textbf{total footprint} of the point to be 
    \begin{equation*}
        T((e_1,e_2),(e_1',e_2'))=\max\{(e_1+1)(e_2+1), (e_1'+1)(e_2'+1)\}. 
    \end{equation*}
    Given a set $\cF \supseteq \cD$, a \textbf{minimal $\cF$-point} is a pair of exponents $((e_1,e_2),(e_1',e_2'))\in \cF$ such that $T((e_1,e_2),(e_1',e_2')) \leq T ((a,b),(c,d))$ for any $((a,b),(c,d)) \in \cF$. We will denote the footprint of any minimal $\cF$-point to be $T^*$. \
\end{definition}
 \begin{remark}\label{r:Fproxy}
     We introduce the set $\cF$ as a proxy for $\cD$, since $\cD$ is hard to fully characterize (note that the conditions in Proposition \ref{p:orthogconditions} are sufficient but not necessary). Given two sets $\cF_X\supseteq\cD_X,\cF_Y \supseteq \cD_Y$, we will assume that 
     \[
     \cF=\{((e_1,e_2),(e_1',e_2')) : (e_1,e_1') \in \cF_X \text{ and } (e_2,e_2') \in \cF_Y\}.
     \]
     
 \end{remark}

\begin{remark}\label{r:tlessthanTstar}
    So long as $t\leq T^*$, we can ensure that no pair of monomials from $\Delta_t$ lies in $\cD$, meaning that our GMC code will be Hermitian self-orthogonal. Therefore, our choice of GRS code for the $Y$ variable should maximize $T^*$. 
\end{remark}

In the following sections, we will describe our choice for the second GRS code, along with justification that this choice maximizes $T^*$.

\section{Choosing  2 points in the second variable}\label{s:2points}

We continue the notation of the previous section. With the set $P_X$ and vector $v$ fixed, it remains to choose the set $P_Y$ and the vector $w$.  In this section we focus on the case with $n_Y=|P_Y|=2$, from which it follows that $\cD_Y \subseteq \{0,1\}\times \{0,1\}$. The exact elements of $\cD_Y$ will depend on our choice of $P_Y$ and $w$. 

To motivate our selection, first suppose that $\cF_Y=\cD_Y=\{0,1\}\times \{0,1\}$. In this case, the minimal $\cF$-point must have the form $((e_1,0),(e_1',0))$ with $(e_1,e_1') \in \cF_X$. It then follows from Lemma \ref{l:latticemovement} that the minimal $\cF$-point is $((T_1,0),(T_2,0))$, with corresponding total footprint $T_2+1$. 

In choosing $P_Y$ and $w$, it is therefore our priority to remove $(0,0)$ from the set $\cD_Y$. To that end, we let $y_0,y_1$ be arbitrary distinct elements of $\F_{q^2}$.  We choose the vector $w$ so that $w_0^{q+1}+w_1^{q+1}=0$. It then follows that 
\[
\left(ev_{w,P_Y}\left(Y^{0}\right)\cdot_h ev_{w,P_Y}\left(Y^{0}\right) \right)=w_0^{q+1}(y_0)^{0+0q}+w_1^{q+1}(y_1)^{0+0q}=0.
\]

\begin{definition}\label{d:GMC2points}
    With $P_X,v,P_Y,w$ as defined in Section \ref{s:GMCframework} and in the preceding paragraphs, and with $t \leq T^*,$ we define the GMC code over $\fqs$:
    \[
    C:=GMC_t(Q,P_X\times P_Y)
    \]
    where $Q(i,j)=v_i\cdot w_j$.
\end{definition}
Under these conditions, we can take $\cF_Y=\{(1,0),(0,1),(1,1)\} \supseteq \cD_Y$.   At the end of the section we will show that this selection of $P_Y$ and $w$ maximizes $T^*$. Let us now calculate a minimal $\cF$-point under this assumption. 

\begin{lemma}\label{l:fpWLOG}
    There is a minimal $\cF$-point of the form $ ((e_1,1),(e_1',0))$ with $(e_1.e_1') \in \cF_X$ and $e_1<e_1'$. 
\end{lemma}
\begin{proof}
    Using Remark \ref{r:Fproxy}, all $\cF$-points are of the form $((e_1,b),(e_1',d))$ with $b,d\in \{0,1\}$ and at least one of $b,d$ non-zero. Swapping $b$ and $d$ also results in a valid $\cF$-point, and the total footprint can only be reduced by setting one of $b,d$ equal to zero. Using the convention in Remark \ref{r:Xsymmetry} that $e_1<e_1'$ we observe that $T( (e_1,1),(e_1',0)) <T((e_1,0),(e_1',1))$ and the result follows. 
    \end{proof}
    Using this result and Equation \ref{eq:writingXpoint}, we can restrict our search of a minimal $\cF$-point to the set of points 
    \begin{equation}\label{eq:enumFpoints2points}
        ((T_1+1+i\lambda/2-j\tau/2,1), (T_2+1+i\lambda/2+j\tau/2,0))
    \end{equation}
 
    where $i,j \ge 0 $. There are restrictions on the values of $i,j$, coming from the fact that exponent pairs must consist of non-negative integers. We will deal with this during the analysis. Our task is to determine which (permitted) values of $i,j$ result in a minimal $\cF$-point. 
    \begin{definition}
        We define the functions
        \begin{equation*}
            \begin{split}
 F_1(i,j)&=2(T_1+1+i\lambda/2-j\tau/2), \\
 F_2(i,j)&=T_2+1+i\lambda/2+j\tau/2,\\
F(i,j)&=\max\{F_1(i,j),F_2(i,j)\}.\\
            \end{split}
        \end{equation*}

    \end{definition}
\begin{remark}
    The function $F(i,j)$ enumerates all possible total footprints from our candidates for the minimal $\cF$-point in Equation \ref{eq:enumFpoints2points}. Therefore the minimal value of $F(i,j)$ is equal to $T^*$. By determining which permitted values of $i,j$ that achieve this minimum, we will also be able to compute a minimal $\cF$-point and then also $T^*$. 
\end{remark}
Note that $F(0,0)=\max \{2(T_1+1), T_2+1\}$. For convenience, we set the following notation: 
\begin{definition}\label{d:2pointsCD}
    We define the constants
    \[
    C_0=T_1+1,D_0=T_2+1,
    \]
    where $T_1,T_2$ are as given in Table \ref{tab:valueL}. 
\end{definition}
With this notation, observe that $F(0,0)=\max\{2C_0,D_0\}$, and it then follows that 
\[
F_1(i.j)=2(C_0+i\lambda/2-j\tau/2), ~~~F_2(i,j)=D_0+i\lambda/2+j\tau/2.
\]
\begin{lemma}\label{l:2pvaluei}
If $i \ge1$, then $F(i,j)>F(0,0) ~\forall j ~ \ge0$. 
\end{lemma}
\begin{proof}
First note that for $i\ge1$, we have that
\[
F(i,j) \ge F_2(i,j)=D_0+i\lambda/2+j\tau/2\ge D_0+\lambda/2. 
\]
For both Case 1 and Case 3 from Table \ref{tab:valueL}, is easy to check that $D_0+\lambda/2 > \max\{2C_0,D_0\}=F(0,0)$ and the result follows.

Now suppose that we are in Case 2 from Table \ref{tab:valueL}, so that $C_0=\lambda$ and $D_0=\lambda+\tau$. If $\tau$ is even, then by parity, we must have $i \geq 2$. Then 
\[
F(i,j) \ge D_0+\lambda>\max\{2C_0,D_0\}=F(0,0). 
\]
Finally, suppose that $\tau$ is odd. If $i\ge2$ then for the same reason as above we have that $F(i,j)>F(0,0)$. If $i=1$, then by parity we must have $j \geq 1$. Now, we have that 
\[
F_1(1,j)=3\lambda-j\tau, 
F_2(1,j)=3\lambda/2+\tau+j\tau/2.
\]
Since $F_2(1,j)$ is increasing, it is sufficient to show that  $\max\{(F_1(1,j),F_2(1,j)\}>F_1(0,0) ~\forall j.$ We have that $F_1(1,j) \leq F_1(0,0) \iff j\ge \lambda/\tau$, which means that 
\[
F_2(1,j)\geq 2\lambda+\tau >2\lambda = F_1(0,0),
\]
 and the result follows. 
\end{proof}
\begin{remark}
    In other words, the minimum of the function $F(i,j)$ occurs with $i=0$. 

\end{remark}

Next, we will discuss how to find the value of $j$ that minimizes the function $F(0,j)$. This amounts to the comparison of the following constants: 
\begin{definition}\label{d:2pointks}
    We define the following constants when $\tau$ is even:
\begin{equation*}
\begin{split}
        k_1=&\max\left\{\left\lfloor \frac{2C_0-2}{\tau}\right\rfloor,0 \right\},\\
        k_2=&\max\left\{\left\lceil\frac{4C_0-2D_0}{3\tau} \right\rceil,0\right\},\\
        k_3= &\max\left\{\left\lfloor\frac{4C_0-2D_0+2\tau}{3\tau}\right\rfloor,0\right\},
\end{split}
\end{equation*}
and when $\tau$ is odd: 
\begin{equation*}
\begin{split}
        k_1=&\max\left\{\left\lfloor \frac{C_0-1}{\tau}\right\rfloor,0 \right\},\\
        k_2=&\max\left\{\left\lceil\frac{2C_0-D_0}{3\tau} \right\rceil,0\right\},\\
        k_3= &\max\left\{\left\lfloor\frac{2C_0-D_0+2\tau}{3\tau}\right\rfloor,0\right\}.
\end{split}
\end{equation*}
In both cases, we define 
 \begin{equation*}
    ~~~C_j=C_0-j\tau/2, ~~~D_j=D_0+j\tau/2 , 
 \end{equation*}
 where $C_0,D_0$ are as defined in Definition \ref{d:2pointsCD}. 
\end{definition}

\begin{remark}\label{r:CDtoExponents}
    By Lemma \ref{l:2pvaluei}, the pairs monomials of interest are of the form $X^{C_j-1}Y^1, X^{D_j-1}Y^0$, with corresponding total footprint $\max\{2C_j,D_j\}=F(0,j)$. Our goal is to find the value of $j$ that minimizes this total footprint, which will also tell us a minimal $\cF$-point. 
\end{remark}
 Let us now explain what each of these constants means. Note that $F(0,j)=\max\{2C_j,D_j$\}.
 
 The first restriction on the value of $j$ is that of parity, since our monomial exponents must be whole numbers. Since we are only interested in the case with $i=0$, the only actual restriction is that when $\tau$ is odd, $j$ must be even. We will only provide an analysis of the case with $\tau$ even; the other case is solved similarly.

The next requirement is that the exponents in our monomials be non-negative. In terms of $C_j$ and $D_j$, this is equivalent to stating that $C_j \ge 1$; in other words, we require $C_0-j\tau/2\ge 1$, or $ j \leq \frac{2C_0-2}{\tau} $. Therefore, we are searching for $j$ in the range $0 \leq j \leq k_1$. Another consideration is that a monomial $(e_1,e_2) \in E$ must satisfy $e_1 < |P_X|=\lambda\tau\sigma$. This means that we must enforce $D_j \leq\lambda\tau\sigma$, but this is already true if $j\leq k_1$.
 
In order to compute the minimum of the function $F(0,j)$, observe that the sequence $2C_j$ is decreasing, and the sequence $D_j$ is increasing. Since we are interested in the maximum of the two sequences, the value of $j$ that minimizes $F(0,j)$ will be near the crossover point, which is described precisely by $k_2$. 

\begin{lemma}\label{l:2pointk2}
    If $\tau$ is even, then function $F(0,j)$ satisfies 
        \[
    F(0,j)=
    \begin{cases}
        2(C_0-j\tau/2) & \text{if } j <k_2,\\
        D_0+j\tau/2 & \text{if } j\ge k_2. 
    \end{cases}
    \]
\end{lemma}
\begin{proof}
    This follows from the definition of $C_j$ and $D_j$, and the fact that $2C_j \leq D_j$ $\iff j \ge \frac{4C_0-2D_0}{3\tau}$.
\end{proof}

\begin{remark}
    By Lemma \ref{l:2pointk2} it follows that the value that minimizes $F(0,j)$ is either $j=k_2$ or $j=k_2-1$, so long as $k_2\leq k_1$. So, it remains to compare the values of $D_{k_2}$ and $2C_{k_2-1}$. 
\end{remark}
\begin{lemma}\label{l:2pointk3}
    For $0 \leq j \leq k_3$, we have that $2C_{j-1}\ge D_j$. For $j > k_3$, we have that $D_j>2C_{j-1}$. 
\end{lemma}
\begin{proof}
    Using the definition of $C_{j-1}$ and $D_j$, we can see that $2C_{j-1}\ge D_j\iff j \leq \frac{4C_0-2D_0+2\tau}{3\tau} $.
\end{proof}
\begin{remark}
    We note that either $k_3=k_2$ or $k_3=k_2-1$
\end{remark}

We summarize these results in the following proposition. 
\begin{proposition}\label{p:Tstar2points}
If $\tau$ is even then
\[
T^*=
\begin{cases}
    F(0,k_1)=2C_{k_1}&  \text{if } k_1<k_2,\\
    F(0,k_2)=D_{k_2}& \text{if } k_3=k_2\leq k_1,\\
    F(0,k_2-1)=2C_{k_2-1} & \text{if }k_3<k_2 \leq k_1.
\end{cases}
\]
If $\tau$ is odd then 
\[
T^*=
\begin{cases}
    F(0,2k_1)=2C_{2k_1}&  \text{if } k_1<k_2,\\
    F(0,2k_2)=D_{2k_2}& \text{if } k_3=k_2\leq k_1,\\
    F(0,2k_2-2)=2C_{2k_2-2} & \text{if }k_3<k_2 \leq k_1,
\end{cases}
\]
where $k_1,k_2,k_3$ and $C_j,D_j$ are as given in Definition \ref{d:2pointks}. 
\end{proposition}
\begin{proof}
Suppose that $\tau$ is even. 
    If $k_2>k_1$, then it follows from Lemma \ref{l:2pointk2} that $F(0,j)$ is strictly decreasing for all $0\leq j \leq k_1$, thus the minimum occurs at $F(0,k_1)=2C_{k_1}$. 

    If $k_2 \leq k_1$, then by Lemma \ref{l:2pointk2} the minimum is either $2C_{k_2-1} $ or $D_{k_2}$. If $k_3=k_2$, then in particular $k_2\leq k_3$, which means that $2C_{k_2-1}\ge D_{k_2}$ by Lemma \ref{l:2pointk3}. Thus the minimum is given by $D_{k_2}$. If $k_3=k_2-1$, then the minimum will be at $F(0,k_2-1)=2C_{k_2-1}$. 

    The proof is similar when $\tau$ is odd. 
\end{proof}

\begin{remark}
    Given the parameters $\lambda,\tau,\rho,\sigma$ of the code construction, is it now straightforward to compute $T^*$, and also therefore the largest minimum distance so that the code is Hermitian self-orthogonal. We give some examples of resulting code families in Section \ref{s:families}. 
\end{remark}

\begin{theorem}\label{t:codes2points}
    Let $q\ge 4$ be a prime power. Let $\lambda>1$ be a divisor of $q-1$ and let $\tau,\rho>1$ be divisors of $q+1$. We assume that $\gcd(\lambda,\tau)=1$ and that $\frac{\rho}{\gcd(\lambda\tau,\rho)}\ge2.$ Let $\sigma$ be any integer with $2 \leq \sigma \leq \frac{\rho}{\gcd(\lambda\tau,\rho)}$. Let $n=2\lambda\tau\sigma$, and let $T^*$ be given by Proposition \ref{p:Tstar2points}. 
Then for any $d$ with  $2\le d \le T^*$ there exists a 
$$\left[\left[n,n-2\left(d-1+\left\lfloor\frac{d-1}{2}\right\rfloor\right),d\right]\right]_q$$ 
quantum stabilizer  code. 
\end{theorem}
\begin{proof}
    We consider the GMC code constructed in Definition \ref{d:GMC2points}. The length of the code is $n=|P_X\times P_Y|=2\lambda\tau\sigma$. If $t \leq T^*$, then by Remark \ref{r:tlessthanTstar} our GMC code will be Hermitian self-orthogonal. The parameters of the corresponding quantum code follow from Corollary \ref{c:GMCCParams} and Lemma \ref{l:SizeDelta}. 
\end{proof}
In the final part of this section, we show that no other selection for $P_Y,w$ can increase $T^*$. Recall that in order to remove the point $(0,0)$ from the set $\cD_Y$, it was necessary to introduce the constraint: 
\[
w_0^{q+1}+w_1^{q+1}=0
\]
Under this condition, the minimal $\cF$-point is the form $((e_1,1),(e_1',0))$. If a different selection of $P_Y,w$ were to increase $T^*$, it would have to remove both $(0,0)$ and $(1,0)$ from $\cD_Y$. However, there is no choice of vector or evaluation points that can achieve this. To see this, let $P_Y=\{y_0,y_1\}$ be the $Y$-evaluation points. Then the condition $(1,0) \notin \cD_Y$ amounts to 
\[
w_0^{q+1}y_0^{1+0q}+w_1^{q+1}y_1^{1+0q}=0
\]
which together with the previous condition implies that $y_0=y_1$, a contradiction. 

Thus, we can see that when $|P_Y|=2$, the choice of evaluation points can be arbitrary, and that our chosen coordinate vector provides the greatest range of parameters for which our GMC code will be Hermitian self-orthogonal. It is possible to choose $w$ so that $(1,1) \notin \cD_Y$, but this would not increase $T^*$. This could be of use in entanglement-assisted quantum codes, as such a property would reduce the size of the Hermitian hull, but we do not consider such codes in this paper. 

\section{Choosing $n_Y$ points in the second variable. }

We will return to the notation of Section \ref{s:GMCframework}, and redefine the objects of the previous section in a more general setting. Recall that the set $P_X$ and vector $v$ are fixed, and we must choose the set $P_Y$ and the vector $w$. 

In this section, we consider $P_Y=\{y_0,\ldots y_{n_Y-1}\}\subseteq \F_q$. The reason for choosing points in $\F_q$ rather than $\fqs$ will be explained shortly. With this many points, the $Y$ monomials under consideration are $Y^0,Y^1, \ldots, Y^{n_Y-1}$, so $\cD_Y \subseteq \{0,\ldots,n_Y-1\}^2 $.

We choose our coordinate vector in the following way.  Consider the $(n_Y-1)\times n_Y$ matrix over $\F_q$:
\[
\begin{bmatrix}
    1 &1& \ldots& 1\\
    y_0 &y_1 &\ldots & y_{n_Y-1}\\
    \vdots & \vdots & \ddots &\vdots\\
    y_0^{n_Y-2}&y_1^{n_Y-2}& \ldots & y_{n_Y-1}^{n_Y-2}
\end{bmatrix}.
\]
We consider the solutions over $\F_q$ and choose some non-zero element of the kernel $w_q\in \F_q^{n_Y}$. Since the matrix is a Vandermonde matrix, it must be that every coordinate of $w_q$ is non-zero. We then choose our coordinate vector $w \in (\F_{q^2}^*)^{n_Y}$ so that $w(j)^{q+1}=w_q(j)$. We can now examine the Hermitian inner product: 

\[
\left(ev_{w,P_Y}\left(Y^{e_2}\right)\cdot_h ev_{w,P_Y}\left(Y^{e_2'}\right) \right)=\sum_jw_j^{q+1}(y_j)^{e_2+qe_2'} = \sum_jw_j^{q+1}(y_j)^{e_2+e_2'} .
\]
where the second equality uses the fact that $y_j^q=y_j$ since $P_Y\subseteq \F_q$. Thus, if $e_2+e_2' \leq n_Y-2$ then the Hermitian inner product will be zero. Therefore, we consider $\cF_Y=\{(e_2,e_2'):e_2+e_2'>n_Y-2\}$, and it is clear that $\cF_Y \supseteq \cD_Y$. It follows that
     \[
     \cF=\{((e_1,e_2),(e_1',e_2')): (e_1,e_1')\in\cF_X \text{ and } e_2+e_2'>n_Y-2\}.
     \]
Let us now come to the task of computing a minimal $\cF$-point with respect to the total footprint. For a given $e_2'$, we can minimize the footprint by choosing $e_2$ to be minimal with respect to the constraint $e_2+e_2>n_Y-2$. Therefore, the minimal $\cF$-point is of the form 
$(e_1,n_Y-1-l,e_1',l) $ with $(e_1,e_1') \in \cF_X$ and $0 \leq l \leq n_Y-1$. Without loss of generality we can assume that $e_1<e_1'$ and $l \leq (n_Y-1)/2$, by the same reasoning as Lemma \ref{l:fpWLOG}.

Recall that any point $(e_1,e_1')\in \cF_X$ can be written as 
\[
(e_1,e_1') = (T_1,T_2)+i(\lambda/2,\lambda/2)+j(-\tau/2,\tau/2).
\]
Therefore, we can define the following function
\[
F(i,j,l)=T((e_1,n_Y-1-l),(e_1',l)).
\]

Computing the value of $T^*$ therefore amount to minimizing the function $F(i,j,l)$. As before, we can define set $C_0=T_1+1,D_0=T_2+1$ and define
\[
 F_1(i,j,l)=(n_Y-l)(e_1+1), ~~~~F_2(i,j,l)=(l+1)(e_1'+1),
\]
so that 
\[
F(i,j,l)=\max\{F_1(i,j,l),F_2(i,j,l)\}. 
\]

\begin{lemma}\label{l:RpointsValuei}
    If $n_Y \leq\frac{D_0+\lambda}{C_0}$, then the minimum of $F(i,j,l)$ occurs with $i=0$. 
\end{lemma}
\begin{proof}
    It is sufficient to show that if $i \ge1$ and $l$ is fixed, then $F(i,j,l)>F(0,0,l)=\max\{F_1(0,0,l),F_2(0,0,l)\}$ for any choice of $j,$.  Since $F_2(i,j,l)$ is a strictly increasing function (with respect to $j$), we only have to show that $F(i,j,l)>F_1(0,0,l)$. 
    We have that 
    \begin{equation*}
        \begin{split}
            F_1(i,j,l)=&(n_Y-l)(C_0+i\lambda/2-j\tau/2),\\
            F_2(i,j,l)=&(l+1)(D_0+i\lambda/2+j\tau/2),\\
            F_1(0,0,l)=&(n_Y-l)(C_0).
        \end{split}
    \end{equation*}

Now, it is easy to check that $F_1(i,j,l)\leq F_1(0,0,l)$ $\iff$ $j \geq (i\lambda)/\tau$. It then follows that 
\begin{equation*}
    \begin{split}
        F_2(i,j,l)&=(l+1)(D_0+i\lambda/2+j\tau/2)\\
        &\ge D_0+i\lambda/2+j\tau/2 \\
        &\ge D_0+i\lambda/2+i\lambda/2  \\
        &\ge n_Y(C_0)\\
        &\ge F_1(0,0,l)
    \end{split}
\end{equation*}

Thus, for every value of $j,l$ we have that $\max \{F_1(i,j,l),F_2(i,j,l)\} \ge F(0,0,l) $, and the result follows. 

\end{proof}
\begin{remark}
    It is worth pointing out that for any choice of $\lambda$ and $\tau$ we have that $(D_0+\lambda)/C_0\ge2$, so this restriction on $n_Y$ never results in a trivial construction. 
\end{remark}

\begin{lemma}\label{l:RpointsValuel}
    If $n_Y \leq \frac{D_0+\lambda}{C_0}$ and $D_0\ge \lambda$, then the minimum of $F(i,j,l)$ occurs with $l=0$.  
\end{lemma}
\begin{proof}
    It is sufficient to show that if $l \ge 0$ then for any $j$, we have that $F_2(0,j,l)\ge F(0,0,0)$. It is clear that $F_2(0,j,l)\ge 2D_0 \ge D_0=F_2(0,0,0)$, and it follows from the assumptions that $F_2(0,j,l) \ge 2D_0\ge D_0+\lambda \ge n_YC_0$. Therefore $F_2(0,j,l) \ge \max\{F_1(0,0,0),F_2(0,0,0)\}=F(0,0,0)$. 
\end{proof}

    As a consequence of Lemmas \ref{l:RpointsValuei} and \ref{l:RpointsValuel}, when considering the minimum of the function $F(i,j,l)$ we need only consider the values $F(0,j,0)$. Let use set now proceed to find the value of $j$ which minimizes this function. Just as in Section \ref{s:2points}, this amounts to a comparison of the following constants:

    \begin{definition}\label{d:Rpointks}
    We define the following constants when $\tau$ is even:
\begin{equation*}
\begin{split}
        k_1=&\max\left\{\left\lfloor \frac{2(C_0-1)}{\tau}\right\rfloor,0 \right\}\\
        k_2=&\max\left\{\left\lceil\frac{2(n_YC_0-D_0)}{(n_Y+1)\tau} \right\rceil,0\right\}\\
        k_3= &\max\left\{\left\lfloor\frac{2(n_YC_0-D_0)+n_Y\tau}{(n_Y+1)\tau}\right\rfloor,0\right\}
\end{split}
\end{equation*}
and when $\tau$ is odd: 
\begin{equation*}
\begin{split}
        k_1=&\max\left\{\left\lfloor \frac{C_0-1}{\tau}\right\rfloor,0 \right\},\\
        k_2=&\max\left\{\left\lceil\frac{n_YC_0-D_0}{(n_Y+1)\tau} \right\rceil,0\right\},\\
        k_3= &\max\left\{\left\lfloor\frac{n_YC_0-D_0+n_Y\tau}{(n_Y+1)\tau}\right\rfloor,0\right\},
\end{split}
\end{equation*}
 where $C_0,D_0$ are as defined in Definition \ref{d:2pointsCD}. 
\end{definition}
We only provide analysis for the case where $\tau$ is even. As explained after Definition \ref{d:2pointks} in Section \ref{s:2points}, we are searching for $j$ in the range $0 \leq j \leq k_1$. 
\begin{lemma}\label{l:Rpointk2}
    If $\tau$ is even, then function $F(0,j)$ satisfies 
        \[
    F(0,j,0)=
    \begin{cases}
        n_Y(C_0-j\tau/2) & \text{if } j <k_2,\\
        D_0+j\tau/2 & \text{if } j\ge k_2. 
    \end{cases}
    \]
\end{lemma}
\begin{proof}
    This follows from the definition of $C_j$ and $D_j$, and the fact that $n_YC_j \leq D_j$ $\iff j \ge \frac{2(n_YC_0-D_0)}{(n_Y+1)\tau}$.
\end{proof}

\begin{lemma}\label{l:Rpointk3}
    For $0 \leq j \leq k_3$, we have that $n_YC_{j-1}\ge D_j$. For $j > k_3$, we have that $D_j>n_YC_{j-1}$. 
\end{lemma}
\begin{proof}
    Using the definition of $C_{j-1}$ and $D_j$, we can see that $n_YC_{j-1}\ge D_j\iff j \leq \frac{2(n_YC_0-D_0)+n_Y\tau}{(n_Y+1)\tau} $.
\end{proof}
We summarize these results in the following proposition. 
\begin{proposition}\label{p:TstarRpoints}
If $\tau$ is even then
\[
T^*=
\begin{cases}
    F(0,k_1,0)=n_YC_{k_1}&  \text{if } k_1<k_2,\\
    F(0,k_2,0)=D_{k_2}& \text{if } k_3=k_2\leq k_1,\\
    F(0,k_2-1,0)=n_YC_{k_2-1} & \text{if }k_3<k_2 \leq k_1.
\end{cases}
\]
If $\tau$ is odd then 
\[
T^*=
\begin{cases}
    F(0,2k_1,0)=n_YC_{2k_1}&  \text{if } k_1<k_2,\\
    F(0,2k_2,0)=D_{2k_2}& \text{if } k_3=k_2\leq k_1,\\
    F(0,2k_2-2,0)=n_YC_{2k_2-2} & \text{if }k_3<k_2 \leq k_1,
\end{cases}
\]
where $k_1,k_2,k_3$ and $C_j,D_j$ are as given in Definition \ref{d:Rpointks}. 
\end{proposition}
\begin{proof}
Suppose that $\tau$ is even. 
    If $k_2>k_1$, then it follows from Lemma \ref{l:Rpointk2} that $F(0,j,0)$ is strictly decreasing for all $0\leq j \leq k_1$, thus the minimum occurs at $F(0,k_1,0)=n_YC_{k_1}$. 

    If $k_2 \leq k_1$, then by Lemma \ref{l:2pointk2} the minimum is either $n_YC_{k_2-1} $ or $D_{k_2}$. If $k_3=k_2$, then in particular $k_2\leq k_3$, which means that $n_YC_{k_2-1}\ge D_{k_2}$ by Lemma \ref{l:Rpointk3}. Thus the minimum is given by $D_{k_2}$. If $k_3=k_2-1$, then the minimum will be at $F(0,k_2-1,0)=n_YC_{k_2-1}$. 

    The proof is similar when $\tau$ is odd. 
\end{proof}

\begin{theorem}\label{t:codesRpoints}
        Let $q\ge 4$ be a prime power. Let $\lambda>1$ be a divisor of $q-1$ and let $\tau,\rho>1$ be divisors of $q+1$. We assume that $\gcd(\lambda,\tau)=1$ and that $\frac{\rho}{\gcd(\lambda\tau,\rho)}\ge2.$ Let $\sigma$ be any integer with $2 \leq \sigma \leq \frac{\rho}{\gcd(\lambda\tau,\rho)}$. 
    Let $n_Y\leq q$, and suppose that the assumptions of Lemma \ref{l:RpointsValuel} hold. 
        Let $n=n_Y\lambda\tau\sigma$, and let $T^*$ be given by Proposition \ref{p:TstarRpoints}. 
Then for any $d$ with  $2\le d \le T^*$ there exists a 
$$\left[\left[n,n-2|\Delta_d|,d\right]\right]_q$$ 
quantum stabilizer  code. 
\end{theorem}

\begin{remark}
    We provide a formula for computing $|\Delta_d|$ in Lemma \ref{l:SizeDelta}.
\end{remark}

\begin{exmp}\label{exmp:codeexample}
Let us work through an example to see Theorem \ref{t:codesRpoints} in action. Suppose that $q=8, \lambda=7,\tau=3,\rho=9,\sigma=2 $ and $ n_Y=3$. We are in Case 3 of Table \ref{tab:valueL}, thus $C_0=\frac{\lambda+\tau}{2}=5, D_0=C_0+\tau=8$. We can verify that $D_0=8\ge 7=\lambda$ and $n_Y =3 \leq 3 = \frac{D_0+\lambda}{C_0}$. We compute that $k_1=k_2=k_3=1$, and thus by Proposition \ref{p:TstarRpoints} we have that $T^*=F(0,2,0)=D_2=D_0+2(\tau/2)=11$. Thus, by Theorem \ref{t:codesRpoints}, we get quantum codes with parameters $[[126,126-2|\Delta_d|,d]]_8$ for $2 \leq d \leq 11$ over $q=8$. In the next section, we give a formula for $|\Delta_d|$. 

In this case, a minimal $\cF$-point is the pair of monomials $X^1Y^2, X^{10}Y^0$ (see Remark \ref{r:CDtoExponents}). However, since $(1,10)$ satisfies the third condition of Proposition \ref{p:orthogconditions}, the evaluation vectors corresponding to these monomials are actually orthogonal. We can manually check that the GMC code is in fact Hermitian self-orthogonal for any $2 \leq d \leq 15$.
\end{exmp}

\section{The parameters of the codes and Comparison with the GV Bound}

We now give an explicit description of the code parameters which we have constructed. 

\begin{lemma}\label{l:SizeDelta}
    Let $|P_Y|=n_Y$. Then 
    \[
    |\Delta_t| = \sum_{i=1}^{n_Y} \left\lfloor\frac{t-1}{i}\right\rfloor .
    \] 
\end{lemma}
\begin{proof}
    We partition all monomials in the quotient ring into the groups $X^eY^i$ for $0 \leq i \leq n_Y$. For a given $i$, the number of monomials is equal to the number of non-negative integer solutions to the inequality $n_Y(e+1)<t$ which is $\left\lfloor\frac{t-1}{n_Y}\right\rfloor$, and the result follows. 
\end{proof}

\begin{remark}\label{r:deltaTtwopoints}
    
For $n_Y=2$, we have the rather nice formula 
\[
|\Delta_t|=t-1+\left\lfloor\frac{t-1}{2}\right\rfloor.
\]
\end{remark}

Recall that the quantum Singleton bound states that $k+2d\leq n+2$. We can therefore define the \textbf{defect} of an $[[n,k,d]]_q$ quantum code to be $DF=n+2-k-2d$. The defect is zero precisely when we have an MDS code. By Corollary \ref{c:GMCCParams}, our GMC code has parameters $[[n,n-2|\Delta_t|,t]]_q$. Therefore, the defect our our quantum GMC codes when $n_Y=2$ as a function of $t$ is 

\[
DF(t)=2\left\lfloor\frac{t-1}{2}\right\rfloor. 
\]
In particular, our codes are closer to MDS codes when the minimum distance is small.

\subsection{Codes Beating QGV bound}

Let us recall the  quantum Gilbert-Varshamov bound.

\begin{theorem}[Quantum Gilbert-Varshamov Bound]
    Suppose that $n>k \geq 2$, $d\geq 2$, and $n\equiv k \mod 2$. If
    \begin{equation}\label{qgvb}
    \frac{q^{n-k+2}-1}{q^2-1}\geq \sum_{i=1}^{d-1} (q^2-1)^{i-1} {n \choose i}
    \end{equation}
  then  there exists a pure stabilizer quantum code with parameters $[[n,k,d]]_q$.
\end{theorem}

We say that a parameter set $n,k,d, q$ beats the QGV bound if the inequality 
\eqref{qgvb} is not satisfied.
We will demonstrate that infinitely many of the codes constructed in 
this paper beat the QGV bound.
Before our proof we need a preliminary lemma.

\begin{lemma}\label{increasing1}
Let $q\ge 11$.
Let $\alpha$ be a  real number in the interval $[\frac{4}{3q-2},\frac{1}{2}]$.
Then
\[
\frac{1}{\alpha} \ q^{1+\alpha}<q^2-2.
\]
\end{lemma}

\begin{proof}
Let $x$ be a real variable and consider the function
\[
f(x)=\frac{q^{1+x}}{x}
\]
on the interval $[\frac{4}{3q-2},\frac{1}{2}]$.
It is not hard to show that the maximum of $f(x)$ on this interval is at
 $x=\frac{4}{3q-2}$, the left endpoint.
We must show that $f(\frac{4}{3q-2})<q^2-2$.
So we want to prove that
\[
q^{1+\frac{4}{3q-2}}<\frac{4}{3q-2}(q^2-2).
\]
Taking logs, it is equivalent to proving that
\[
(1+\frac{4}{3q-2})\log q <\log \frac{4(q^2-2)}{3q-2}.
\]
Letting
\[
F(q)=\log \frac{4(q^2-2)}{3q-2}-(1+\frac{4}{3q-2})\log q
\]
we want to show that $F(q)>0$ for $q\ge 11$. 
It can be easily checked that $F(q)$ is an increasing function of $q$, and that
$F(11)>0$, which proves that $F(q)>0$ for $q\ge 11$.
\end{proof}

\begin{theorem}\label{t:beatGV2}
Let $q$ be a prime power with $q\ge 11$.
Let $C$ be a quantum stabilizer code constructed in Theorem \ref{t:codes2points}
with $n=2(q^2-1)$, and $d$ in the range $5\le d \le 3q/2$.\\
Then $C$ beats the QGV bound.
\end{theorem}

\begin{proof}

To prove that a code beats the QGV bound we would like to prove that
\[
 \frac{q^{n-k+2}-1}{q^2-1}< \sum_{i=1}^{d-1} (q^2-1)^{i-1} {n \choose i}.
\]
It suffices to prove that 
\[
q^{n-k+2}<(q^2-1)^{d-1}{n \choose d-1}.
\]
In our case, we use the codes whose existence is proved in 
Theorem \ref{t:codes2points}, so we have
\[
n-k+2=2d+2\left\lfloor\frac{d-1}{2}\right\rfloor=
\begin{cases}3d-1 &\text{if $d$ odd}\\ 3d-2 &\text{if $d$ even}
\end{cases}
\]
So it suffices to prove that 
\[
q^{3d-1}<(q^2-1)^{d-1}{n \choose d-1}.
\]
In our case $n=2(q^2-1)$, and using the estimate  $ (\frac{n}{k})^k< {n \choose k}$
it suffices to prove that 
\[
q^{3d-1}<\biggl( \frac{2(q^2-1)^2}{d-1}\biggr)^{d-1}.
\]
Taking the $d-1$ root of both sides, this inequality becomes
\[
q^{3+\frac{2}{d-1}}< \frac{2(q^2-1)^2}{d-1}.
\]
Let $\alpha=\frac{2}{d-1}$, it suffices for us to prove that
\[
\frac{1}{\alpha} q^{3+\alpha}<(q^2-1)^2.
\]
Since $(q^2-1)^2=q^4-2q^2+1$ it suffices for us to prove that
\[
\frac{1}{\alpha} q^{3+\alpha}<q^4-2q^2
\]
or equivalently
\[
\frac{1}{\alpha} q^{1+\alpha}<q^2-2.
\]
The truth of this inequality follows from Lemma \ref{increasing1}.
We note that $\alpha=\frac{1}{2}$ when $d=5$, and $\alpha$ decreases as $d$ increases.
The value $d=\frac{3q}{2}$ corresponds to $\alpha=\frac{4}{3q-2}$.
Choosing $\alpha$  in the interval $[\frac{4}{3q-2},\frac{1}{2}]$
 corresponds to choosing $d$ with $5\le d \le \frac{3q}{2}$.

 \begin{remark}
     It is worth pointing out there there are infinitely many codes from Theorem \ref{t:codes2points} satisfying the assumptions of Theorem \ref{t:beatGV2}. For example, if we take the family of codes from Corollary \ref{c:fam1} and set $\sigma=3$ instead of $2$, then we have codes of length $2(q^2-1)$ and minimum distance $2 \leq d \leq \frac{4q+1}{3}$ for every $d$ in this interval
     (and for every $q\ge 11$ satisfying the hypotheses of Corollary \ref{c:fam1}, of which there are infinitely many). An example of such codes for $q=32$ is provided in Table \ref{tab:codesq32}.
 \end{remark}
\end{proof}

\section{Families of Codes}\label{s:families}

We now give some explicit families of quantum stabilizer codes using our construction. We construct three families of codes with length greater than $q^2+1$ and distance greater than $q+1$. 

\begin{corollary}\label{c:fam1}
    Let $q>2$ be  even with $ q \equiv 2 \text{(mod 3)}$,  Then for any $2 \leq d \leq \frac{4q-2}{3}$, there exists a 
    \[
    \left[\left[\frac{4}{3}\left(q^2-1\right), n-2\left(d-1+\left\lfloor\frac{d-1}{2}\right\rfloor\right), d\right]\right]_q
    \]
    quantum stabilizer code. 
\end{corollary}
\begin{proof}
We take $\lambda=q-1, \tau=(q+1)/3, \rho=q+1, \sigma=2, n_Y=2$. The parameters of the code follow from Corollary \ref{c:GMCCParams} and Remark \ref{r:deltaTtwopoints}. 

In order to justify that our codes are Hermitian self-orthogonal for $2 \leq d \leq \frac{4q-2}{3}$, we can apply Proposition \ref{p:Tstar2points}. We are in Case 3 from Table \ref{tab:valueL}, and can compute that $C_0=\frac{2q-1}{3}, D_0=q, k_1=k_2=1, k_3=0$. We can check that  $n_Y=3=(D_0+\lambda) /C_0$ and $D_0=q >q-1 =\lambda$, thus the assumptions of Lemmas \ref{l:RpointsValuei} and \ref{l:RpointsValuel} hold. Therefore, by Proposition \ref{p:Tstar2points}, $T^*=2C_0=\frac{4q-2}{3}$ and the result follows from Theorem \ref{t:codes2points}. 
\end{proof}

Examples of new codes from Corollary \ref{c:fam1} with $q=8$ are listed in Table \ref{tab:CodesFirstFamily}. 
Some of these beat the codes listed in codetables.

\begin{table}[H]
% title of Table
\centering
%\begin{center}
\begin{tabular}{|c|c|c|c|}
\hline
$q,\lambda,\tau,\rho,\sigma,n_Y$&Code& Beats QGV& Comment\\
 \hline 
  \hline $ 8,7,3,9,2,2$ & $[[84,82,2]]_{8}$ & Yes & Matches codetables.de \\ 
 \hline $ 8,7,3,9,2,2$ & $[[84,78,3]]_{8}$ & No&Matches codetables.de \\ 
 \hline $ 8,7,3,9,2,2$ & $[[84,76,4]]_{8}$ & Yes&Matches codetables.de \\ 
 \hline $ 8,7,3,9,2,2$ & $[[84,72,5]]_{8}$ & Yes&Matches codetables.de \\ 
 \hline $ 8,7,3,9,2,2$ & $[[84,70,6]]_{8}$ & Yes &Matches codetables.de\\
 \hline $ 8,7,3,9,2,2$ & $[[84,66,7]]_{8}$& Yes & Beats $[[84,66,6]]_8$ from codetables.de\\ 
 \hline $ 8,7,3,9,2,2$ & $[[84,64,8]]_{8}$& Yes & Beats $[[84,64,7]]_8$ from codetables.de\\ 
 \hline $ 8,7,3,9,2,2$ & $[[84,60,9]]_{8}$& Yes & Beats $[[84,60,8]]_8$ from codetables.de\\ 
 \hline $ 8,7,3,9,2,2$ & $[[84,58,10]]_{8}$& Yes & Beats $[[84,58,8]]_8$ from codetables.de\\ 
 \hline
\end{tabular}

\caption{Some codes from Corollary \ref{c:fam1}}
\label{tab:CodesFirstFamily}
\end{table}

Next we present a family for $q \equiv 3 \text{ (mod $8$)} $.

\begin{corollary}\label{c:fam2}
    Let $q\ge 11$ with $q \equiv 3 \text{ (mod $8$)} $. Then for any $2 \leq d \leq \frac{5q+1}{4}$ there exists a
    \[
    \left[\left[\frac{3}{2}\left(q^2-1\right), n-2|\Delta_d|, d\right]\right]_q
    \]
    quantum stabilizer code, where 
    \[
    |\Delta_d|=d-1+\left\lfloor\frac{d-1}{2}\right\rfloor+ \left\lfloor\frac{d-1}{3}\right\rfloor.
    \]
\end{corollary}
\begin{proof}
    We take $\lambda=q-1, \tau=\frac{q+1}{4}, \rho=4,  \sigma=2, n_Y=3$. The parameters of our GMC code (and resulting quantum stabilizer code) follows from Corollary \ref{c:GMCCParams}. 

    To show that our GMC codes are Hermitian self-orthogonal, we apply Proposition \ref{p:TstarRpoints}. Since $\lambda=q-1$ is even, we are in Case 1 from Table \ref{tab:valueL}, which means that $C_0=(q-1)/2$ and $D_0=q$. We can verify the assumptions that $n_Y=3 \leq (4q-2)/(q-1)=(D_0+\lambda)/C_0)$ and $D_0=q>q-1=\lambda$. We can calculate that for $q\ge 11$, we have $k_1=k_2=k_3=1$. Therefore, $T^*=D_{2k_2}=\frac{5q+1}{4}$ and the result follows from Theorem \ref{t:codesRpoints}. 
\end{proof}

Examples of codes from Corollary \ref{c:fam2} with $q=11$ are listed in Table \ref{tab:CodesSecondFamily}.

\begin{table}[ht]
% title of Table
\centering
%\begin{center}
\begin{tabular}{|c|c|c|c|}
\hline
$q,\lambda,\tau,\rho,\sigma,n_Y$&Code& Beats QGV& Comment\\
 \hline 
  \hline $ 11,5,6,12,2,3$ & $[[180,178,2]]_{11}$ & Yes & MDS \\ 
 \hline $ 11,5,6,12,2,3$ & $[[180,174,3]]_{11}$ & Yes&Singleton defect is 2\\ 
 % \hline $ 11,5,6,12,2,3$ & $[[180,170,4]]_{11}$ & No &\\ 
 \hline $ 11,5,6,12,2,3$ & $[[180,166,5]]_{11}$ & No &Beats $[[180,164,5]]_{11}$ in \cite{barbero24}\\  
 \hline $ 11,5,6,12,2,3$ & $[[180,164,6]]_{11}$ & Yes&Best known \\ 
 % \hline $ 11,5,6,12,2,3$ & $[[180,158,7]]_{11}$ & No &\\ 
 % \hline $ 11,5,6,12,2,3$ & $[[180,156,8]]_{11}$ & No &\\ 
 % \hline $ 11,5,6,12,2,3$ & $[[180,152,9]]_{11}$ & No &\\ 
 % \hline $ 11,5,6,12,2,3$ & $[[180,148,10]]_{11}$ & No& \\ 
 % \hline $ 11,5,6,12,2,3$ & $[[180,144,11]]_{11}$ & No &\\ 
 % \hline $ 11,5,6,12,2,3$ & $[[180,142,12]]_{11}$ & No &\\ 
 % \hline $ 11,5,6,12,2,3$ & $[[180,136,13]]_{11}$ & No &\\ 
 % \hline $ 11,5,6,12,2,3$ & $[[180,134,14]]_{11}$ & No &\\ 
 \hline
\end{tabular}

\caption{Some codes from Corollary \ref{c:fam2}}
\label{tab:CodesSecondFamily}
\end{table}

Next we present a family for $q \equiv 7 \text{ (mod $8$)} $.

\begin{corollary}\label{c:fam3}
    Let $q \equiv 7 \text{ (mod $8$)} $. Then for any $2 \leq d \leq \frac{5q+1}{4}$ there exists a
    \[
    \left[\left[\frac{3}{2}\left(q^2-1\right), n-2|\Delta_d|, d\right]\right]_q
    \]
    quantum stabilizer code, where 
    \[
    |\Delta_d|=d-1+\left\lfloor\frac{d-1}{2}\right\rfloor+ \left\lfloor\frac{d-1}{3}\right\rfloor.
    \]
\end{corollary}

\begin{proof}
    We take $\lambda=\frac{q-1}{2}, \tau=\frac{q+1}{2}, \rho=8, \sigma=2, n_Y=3$. We are in Case 2 from Table \ref{tab:valueL}. As in the proof of Corollary \ref{c:fam2} we can verify the assumptions of Lemmas \ref{l:RpointsValuei} and \ref{l:RpointsValuel}, compute that $k_1=k_2=k_3=1$, and use Proposition \ref{p:TstarRpoints} to compute that $T^*=\frac{5q+1}{4}$. The result follows from Theorem \ref{t:codesRpoints}. 
\end{proof}
For example, when $q=7$ we get the following codes:

\begin{table}[ht]
% title of Table
\centering
%\begin{center}
\begin{tabular}{|c|c|c|c|}
\hline
$q,\lambda,\tau,\rho,\sigma,n_Y$&Code& QGV&Comment\\
\hline
 \hline $ 7,3,4,8,2,3$ & $[[72,70,2]]_{7}$ & Yes&Matches codetables.de \\ 
 \hline $ 7,3,4,8,2,3$ & $[[72,66,3]]_{7}$ & Yes &Matches codetables.de\\ 
 \hline $ 7,3,4,8,2,3$ & $[[72,62,4]]_{7}$ & No &Matches codetables.de\\ 
 \hline $ 7,3,4,8,2,3$ & $[[72,58,5]]_{7}$ & No&Matches codetables.de \\ 
 \hline $ 7,3,4,8,2,3$ & $[[72,56,6]]_{7}$ & Yes&Matches codetables.de \\ 
 % \hline $ 7,3,4,8,2,3$ & $[[72,50,7]]_{7}$ & No &$[[72,50,8]]_{7}$ known\\ 
 % \hline $ 7,3,4,8,2,3$ & $[[72,48,8]]_{7}$ & No &$[[72,48,9]]_{7}$ known\\ 
 % \hline $ 7,3,4,8,2,3$ & $[[72,44,9]]_{7}$ & No& $[[72,44,10]]_{7}$ known\\ 
 \hline
\end{tabular}
\caption{Some Codes From Corollary \ref{c:fam3}}
\label{tab:CodesThirdFamily}
\end{table}

\section{Examples and Comparison with the Literature}

In Tables \ref{tab:codesq8}-\ref{tab:codesq32}, we present some examples of codes from our construction and state whether they beat the quantum Gilbert-Varshamov bound, along with a comparison to known results. In general, choosing $\lambda$ and $\tau$ to be larger will result in $T^*$ being larger, along with the length. However, choosing $\rho$ to be smaller can increase the range of distances for which the GMC code is self orthogonal, since Condition 3 of Proposition \ref{p:orthogconditions} applies to more points.

\begin{table}[ht]
% title of Table
\centering
%\begin{center}
\begin{tabular}{|c|c|c|c|}

\hline
$q,\lambda,\tau,\rho,\sigma,n_Y$&Code&Beats QGV & Comment\\
 \hline 
 \hline $ 8,7,3,9,3,2$ & $[[126,124,2]]_{8}$ & Yes & Beats $[[127,113,3]]_8$ in yvesedel.de  \\ 
 \hline $ 8,7,3,9,3,2$ & $[[126,120,3]]_{8}$ & Yes & Beats $[[127,113,3]]_8$ in yvesedel.de  \\ 
 \hline $ 8,7,3,9,3,2$ & $[[126,118,4]]_{8}$ & Yes& Beats $[[127,113,3]]_8$ in yvesedel.de \\ 
 \hline $ 8,7,3,9,3,2$ & $[[126,114,5]]_{8}$ & Yes& Beats $[[127,106,5]]_8$ in yvesedel.de \\ 
 \hline $ 8,7,3,9,3,2$ & $[[126,112,6]]_{8}$ & Yes& Beats $[[127,106,5]]_8$ in yvesedel.de  \\ 
 \hline $ 8,7,3,9,3,2$ & $[[126,108,7]]_{8}$ & Yes& Beats $[[128,98,7]]_8$ in yvesedel.de\\ 
 \hline $ 8,7,3,9,3,2$ & $[[126,106,8]]_{8}$ & Yes& Beats $[[127,106,5]]_8$ in yvesedel.de \\ 
 \hline $ 8,7,3,9,3,2$ & $[[126,102,9]]_{8}$ & Yes& Beats $[[127,106,5]]_8$ in yvesedel.de \\ 
 \hline $ 8,7,3,9,3,2$ & $[[126,100,10]]_{8}$ & Yes& Beats $[[127,78,10]]_8$ in yvesedel.de\\ 
 \hline
 \hline $ 8,7,3,9,3,3$ & $[[189,187,2]]_{8}$ & Yes& Beats $[[189,177,3]]_{8}$ in yvesedel.de \\ 
 \hline $ 8,7,3,9,3,3$ & $[[189,183,3]]_{8}$ & Yes& Beats $[[189,177,3]]_{8}$ in yvesedel.de \\ 
 \hline $ 8,7,3,9,3,3$ & $[[189,179,4]]_{8}$ & Yes& Beats $[[189,175,4]]_{8}$ in yvesedel.de  \\ 
 \hline $ 8,7,3,9,3,3$ & $[[189,175,5]]_{8}$ & Yes& Beats $[[189,171,5]]_{8}$ in yvesedel.de\\ 
 \hline $ 8,7,3,9,3,3$ & $[[189,173,6]]_{8}$ & Yes& Beats $[[189,171,5]]_{8}$ in yvesedel.de \\ 
 \hline $ 8,7,3,9,3,3$ & $[[189,167,7]]_{8}$ & No& Beats $[[189,161,7]]_{8}$ in yvesedel.de\\ 
 \hline $ 8,7,3,9,3,3$ & $[[189,165,8]]_{8}$ & Yes& Beats $[[189,157,7]]_{8}$ in yvesedel.de \\ 
 \hline $ 8,7,3,9,3,3$ & $[[189,161,9]]_{8}$ & Yes& Beats $[[189,151,9]]_{8}$ in yvesedel.de\\ 
 \hline $ 8,7,3,9,3,3$ & $[[189,157,10]]_{8}$ & Yes& Beats $[[189,157,7]]_{8}$ in yvesedel.de \\ 
 \hline $ 8,7,3,9,3,3$ & $[[189,153,11]]_{8}$ & No& Beats $[[189,145,9]]_{8}$ in yvesedel.de \\ 
 \hline $ 8,7,3,9,3,3$ & $[[189,151,12]]_{8}$ & Yes& Beats $[[189,137,12]]_{8}$ in yvesedel.de\\ 
 \hline
 % \hline $ 8,7,3,9,3,3$ & $[[189,145,13]]_{8}$ & No \\ 
 % \hline $ 8,7,3,9,3,3$ & $[[189,143,14]]_{8}$ & No \\ 
 % \hline $ 8,7,3,9,3,3$ & $[[189,139,15]]_{8}$ & No \\ 
 % \hline
\end{tabular}
\caption{ Codes we can construct with $q=8$}
\label{tab:codesq8}
\end{table}

\begin{table}[ht]
% title of Table
\centering
%\begin{center}
\begin{tabular}{|c|c|c|c|}
\hline
$q,\lambda,\tau,\rho,\sigma,n_Y$&Code& Beats QGV& Comment\\
 \hline 
 \hline $ 7,3,4,8,2,3$ & $[[72,58,5]]_{7}$ & No & Beats $[[72,56,5]]_7$ in \cite{barbero24} \\ 

  \hline $ 7,3,2,8,4,3$ & $[[72,56,6]]_{7}$ & Yes & Beats $[[72,56,5]]_7$ in \cite{barbero24}\\

   \hline 
   $ 7,3,2,8,3,3$ & $[[54,52,2]]_{7}$ & Yes & Matches codetables.de, beats $[[55,51,2]]_7$ in \cite{Tian2024}.\\
    \hline $ 7,3,4,8,2,4$ & $[[96,90,3]]_{7}$ & Yes & Beats $[[98,82,3]]_7$ in \cite{Tian2024}\\ 
 \hline $ 7,3,4,8,2,4$ & $[[96,86,4]]_{7}$ & Yes &Beats $[[98,82,3]]_7$ in \cite{Tian2024} \\ 
  \hline $ 7,3,2,8,3,6$ & $[[108,102,3]]_{7}$ & Yes & Beats $[[116,100,3]]_7$ in \cite{Tian2024} \\ 
    \hline $ 11,5,4,3,3,3$ & $[[180,166,5]]_{11}$ & No & Beats $[[180,164,5]]_{11}$ in \cite{barbero24}\\ 
     \hline $ 13,3,2,7,4,2$ & $[[48,42,3]]_{13}$ & No & Beats $[[52,42,3]]_{13}$ in \cite{Zhang2023}\\ 
      \hline $ 23,11,2,4,2,2$ & $[[88,82,3]]_{23}$ & No & Beats $[[92,82,3]]_23$ in \cite{Zhang2023}\\ 
      \hline
\end{tabular}

\caption{A sample of codes from Theorem \ref{t:codesRpoints}.}
\label{tab:CodeComparisons}
\end{table}

\begin{table}[ht]
% title of Table
\centering
%\begin{center}
\begin{tabular}{|c|c|c|}
\hline
$q,\lambda,\tau,\rho,\sigma,n_Y$&Code&Beats QGV\\
 \hline
  \hline $ 32,31,11,33,3,2$ & $[[2046,2044,2]]_{32}$ & Yes \\ 
 \hline $ 32,31,11,33,3,2$ & $[[2046,2040,3]]_{32}$ & Yes \\ 
 \hline $ 32,31,11,33,3,2$ & $[[2046,2038,4]]_{32}$ & Yes \\ 
 \hline $ 32,31,11,33,3,2$ & $[[2046,2034,5]]_{32}$ & Yes \\ 
 \hline $ 32,31,11,33,3,2$ & $[[2046,2032,6]]_{32}$ & Yes \\ 
 \hline $ 32,31,11,33,3,2$ & $[[2046,2028,7]]_{32}$ & Yes \\ 
 \hline $ 32,31,11,33,3,2$ & $[[2046,2026,8]]_{32}$ & Yes \\ 
 \hline $ 32,31,11,33,3,2$ & $[[2046,2022,9]]_{32}$ & Yes \\ 
 \hline $ 32,31,11,33,3,2$ & $[[2046,2020,10]]_{32}$ & Yes \\ 
 \hline $ 32,31,11,33,3,2$ & $[[2046,2016,11]]_{32}$ & Yes \\ 
 \hline $ 32,31,11,33,3,2$ & $[[2046,2014,12]]_{32}$ & Yes \\ 
 \hline $ 32,31,11,33,3,2$ & $[[2046,2010,13]]_{32}$ & Yes \\ 
 \hline $ 32,31,11,33,3,2$ & $[[2046,2008,14]]_{32}$ & Yes \\ 
 \hline $ 32,31,11,33,3,2$ & $[[2046,2004,15]]_{32}$ & Yes \\ 
 \hline $ 32,31,11,33,3,2$ & $[[2046,2002,16]]_{32}$ & Yes \\ 
 \hline $ 32,31,11,33,3,2$ & $[[2046,1998,17]]_{32}$ & Yes \\ 
 \hline $ 32,31,11,33,3,2$ & $[[2046,1996,18]]_{32}$ & Yes \\ 
 \hline $ 32,31,11,33,3,2$ & $[[2046,1992,19]]_{32}$ & Yes \\ 
 \hline $ 32,31,11,33,3,2$ & $[[2046,1990,20]]_{32}$ & Yes \\ 
 \hline $ 32,31,11,33,3,2$ & $[[2046,1986,21]]_{32}$ & Yes \\ 
 \hline $ 32,31,11,33,3,2$ & $[[2046,1984,22]]_{32}$ & Yes \\ 
 \hline $ 32,31,11,33,3,2$ & $[[2046,1980,23]]_{32}$ & Yes \\ 
 \hline $ 32,31,11,33,3,2$ & $[[2046,1978,24]]_{32}$ & Yes \\ 
 \hline $ 32,31,11,33,3,2$ & $[[2046,1974,25]]_{32}$ & Yes \\ 
 \hline $ 32,31,11,33,3,2$ & $[[2046,1972,26]]_{32}$ & Yes \\ 
 \hline $ 32,31,11,33,3,2$ & $[[2046,1968,27]]_{32}$ & Yes \\ 
 \hline $ 32,31,11,33,3,2$ & $[[2046,1966,28]]_{32}$ & Yes \\ 
 \hline $ 32,31,11,33,3,2$ & $[[2046,1962,29]]_{32}$ & Yes \\ 
 \hline $ 32,31,11,33,3,2$ & $[[2046,1960,30]]_{32}$ & Yes \\ 
 \hline $ 32,31,11,33,3,2$ & $[[2046,1956,31]]_{32}$ & Yes \\ 
 \hline $ 32,31,11,33,3,2$ & $[[2046,1954,32]]_{32}$ & Yes \\ 
 \hline $ 32,31,11,33,3,2$ & $[[2046,1950,33]]_{32}$ & Yes \\ 
 \hline $ 32,31,11,33,3,2$ & $[[2046,1948,34]]_{32}$ & Yes \\ 
 \hline $ 32,31,11,33,3,2$ & $[[2046,1944,35]]_{32}$ & Yes \\ 
 \hline $ 32,31,11,33,3,2$ & $[[2046,1942,36]]_{32}$ & Yes \\ 
 \hline $ 32,31,11,33,3,2$ & $[[2046,1938,37]]_{32}$ & Yes \\ 
 \hline $ 32,31,11,33,3,2$ & $[[2046,1936,38]]_{32}$ & Yes \\ 
 \hline $ 32,31,11,33,3,2$ & $[[2046,1932,39]]_{32}$ & Yes \\ 
 \hline $ 32,31,11,33,3,2$ & $[[2046,1930,40]]_{32}$ & Yes \\ 
 \hline $ 32,31,11,33,3,2$ & $[[2046,1926,41]]_{32}$ & Yes \\ 
 \hline $ 32,31,11,33,3,2$ & $[[2046,1924,42]]_{32}$ & Yes \\
 \hline
\end{tabular}

\caption{A sample of codes satisfying Theorem \ref{t:beatGV2} for $q=32$}
\label{tab:codesq32}
\end{table}

\section*{Data Availability}
Data sharing is not applicable to this article as no datasets were generated or analyzed
during the current study.

\section*{Conflict of interest}
We declare no conflicts of interest.

\end{document}